\documentclass[smallextended]{svjour3}       

\usepackage[breaklinks,colorlinks=true,linkcolor=blue,citecolor=red,backref=page]{hyperref}

\usepackage[sectionbib,round]{natbib}

\usepackage{amsfonts} 
\usepackage{amssymb,amsmath} 
\usepackage{graphicx,color}
  
\DeclareMathAlphabet{\itbf}{OML}{cmm}{b}{it}

\def\EE{\mathbb{E}}

\def\RR{\mathbb{R}}

\def\eps{\varepsilon}

\newcommand{\ea}{\end{eqnarray}}  
\newcommand{\ba}{\begin{eqnarray}}  
\newcommand{\ee}{\end{equation}}  
\newcommand{\be}{\begin{equation}}  
\newcommand{\ean}{\end{eqnarray*}}  
\newcommand{\ban}{\begin{eqnarray*}}

\newcommand{\dessous}[2]{
\renewcommand{\arraystretch}{0.5} 
\begin{array}[t]{c}
{#1} \\
\scriptstyle
{#2}
\displaystyle
\end{array}
\renewcommand{\arraystretch}{1.0}
}

\smartqed  
\usepackage{graphicx}
 \journalname{Annals of Finance}

\begin{document}  

  \title{Option Pricing under Fast-varying and \\  Rough Stochastic Volatility
}

\titlerunning{Fast-varying and Rough Stochastic Volatility}        

\author{Josselin Garnier        \and
        Knut S\o lna 
}


\institute{J.  Garnier \at
                Centre de Math\'ematiques Appliqu\'ees, \\
Ecole Polytechnique, 91128 Palaiseau Cedex, 
France \\
              \email{ josselin.garnier@polytechnique.edu}           
           \and
           K.  Solna \at
 Department of Mathematics,  \\
University of California, Irvine CA 92697 \\
\email{ ksolna@math.uci.edu}
}

\date{Received: date / Accepted: date}

\maketitle

  \begin{abstract} 
Recent  empirical studies suggest that 
the volatilities  associated with   financial time series  
exhibit short-range correlations. This entails that the volatility
process is very rough and its autocorrelation exhibits sharp decay
 at the origin. Another classic stylistic feature often assumed for the volatility is that
 it is mean reverting.  In this paper it is shown that  the price impact 
 of a rapidly mean reverting rough volatility model coincides with that associated with 
 fast  mean  reverting Markov  stochastic volatility models.   
 This reconciles  the empirical observation of rough volatility paths
 with the good fit of the implied volatility surface to models of fast mean reverting
 Markov  volatilities.
 Moreover, the result conforms with  recent numerical results 
 regarding rough stochastic volatility  models. 
It extends  the scope of models for which
 the asymptotic results of fast mean reverting Markov volatilities are valid. 
The paper concludes  with  a general discussion of fractional volatility asymptotics
 and their interrelation.  The regimes discussed there include fast and slow volatility 
 factors with  strong or  small volatility fluctuations and with
 the limits not commuting in general. 
  The notion of a characteristic term structure
 exponent is introduced, this exponent  governs the implied volatility
 term structure in the various asymptotic regimes. 
\end{abstract}

\section{Introduction}

The assumption that the volatility
is constant, as in the standard Black--Scholes model,  is not realistic.
Indeed, practically, in order to match observed prices, one needs to use an implied volatility that depends on the 
pricing parameters.
Therefore, a consistent parameterization of the implied volatility is needed so that, 
after calibration of the implied volatility model to  liquid contracts,
it can be used for  pricing of less liquid contracts written on the same underlying.
Stochastic volatility models have been introduced because they give such consistent 
parameterizations of the implied volatility.  Here we will consider a specific class
of stochastic volatility models and identify the associated parameterization
of the price correction and  associated  implied volatility correction
that follow from the volatility fluctuations. 
 
 Empirical studies suggest that the volatility may 
exhibit  a ``multiscale'' character as in
\cite{bollerslev,breidt,viens1,cont0,cont,engle,oh}. That is,
correlations that decay  as a power law in  offset rather than as an
exponential function as in a Markov process.
 Recent empirical evidence in particular shows that stochastic volatility 
is often rough  with rapidly decaying correlations
 at the origin (see \cite{gatheral1}). 
In \cite{funahashi} it was shown numerically that the implied volatility correction
for a rough fractional  stochastic volatility model tends to the correction
associated with the Markov case in the regime of long time
to maturity. 
This is consistent with the analytic result derived in this paper where we consider a 
fast mean reverting rough  volatility. 
In this paper, using the martingale method, we get an analytical expression
for the option price and the corresponding implied volatility 
in the regime when the volatility process is fast mean reverting and rough.
The main conclusion is that the corrections to  the option prices and the
corresponding implied volatilities  have exactly the same forms as in the mixing case 
(when the stochastic volatility is Markov).

 A main technical aspect of our derivation
is a careful analysis of the form and properties of the covariation
between the Brownian motion and  the martingale process being the
{conditional square volatility shift}  over the time epoch of interest, see 
Eq. (\ref{def:Kt}) below.
It is   important in this context that we incorporate leverage in our  model so that
these processes are  correlated leading to a non-trivial covariation. 
Another  main aspect of our modeling is that we model the stochastic volatility
fluctuations as being stationary. In a number of recent works a model
for the volatility has been introduced where the initial time plays a special role
leading to a non-stationary process which is artificial from the modeling 
viewpoint. However, we show that in fact in the regime  of  fast mean reversion the asymptotic results of 
the non-stationary case coincide with those of the stationary case considered
here since the volatility process then  forgets its initial state.   
 
  A central aspect of our analysis is the notion of time scales and time scale separation.
 It is then important to identify a reference time scale.    Here we will use the characteristic
 diffusion time
$  \bar\tau = {2}/{{\overline\sigma}^2}$,
 as the reference time, where  $\bar\sigma$ is the effective volatility, see Eqs. (\ref{def:stochmodel}) and (\ref{def:barsigma})  below.  
 Then  we consider a regime where the time to maturity and the characteristic diffusion time are of the same order of magnitude,  
 while the mean reversion time of the volatility fluctuations  is  small relative to  the 
 characteristic  diffusion time. 
 
Note that various asymptotic frameworks can be considered in the context of option pricing
 and we discuss some of them in this paragraph. 
 The choice of the  appropriate  asymptotic framework depends
 on the particular market and contract under consideration. 
 One may consider an asymptotic framework where the volatility fluctuations have small amplitude,
and the time to maturity, the characteristic diffusion time, and
 the mean reversion time for the volatility are of the same order of magnitude. 
 This approach is sensitive to the time dynamic aspects of the volatility. 
 Such an asymptotic situation  was considered for instance in \cite{fouque11} in the mixing
 case and recently in \cite{sv1} in the context of rough volatility factors. 
  However, this asymptotic situation  does not capture
 situations with large volatility fluctuations over the  time scale of the time to maturity.
In the asymptotic framework considered in this paper  the volatility 
 mean reversion time is small relative  to the time to maturity. In such a framework with
 long time horizon for the contract  we can also incorporate
 order one or strong volatility fluctuations over the time scale of the time to maturity. 
 Such an  asymptotic framework is also considered in \cite{25,fouque11}
  in the context of mixing processes.
 Note   that as an option  contract approaches maturity the sensitivity to   
  the payoff function is enhanced leading to important and interesting issues.  
  Indeed a number of recent works  consider implied volatility asymptotics 
 in a regime of short time to maturity,   see for instance \cite{viens3}  and references 
 therein.  The case with contracts that are moreover close to  the money are discussed in 
 \cite{bayer} and \cite{fukasawa15} for instance.   Asymptotics in the context of large strikes
 is discussed in \cite{rodger} for instance.
   We finally remark that for some special models, like the Heston model (see \cite{heston}) it is possible
 to get explicit or semi-explicit option price formulas in a context where volatility fluctuations have large amplitude
 and the mean reversion time is of the same order as the time to maturity.   
 
In this paragraph we discuss some special aspects of long- and short-range correlation properties
in various asymptotic regimes.  
In \cite{sv1} we consider  the situation where the multiscale  stochastic volatility has 
 small amplitude,  and we treat also the case where
 it has  slow variations, which is similar from the analytic viewpoint, 
with the latter case corresponding to a 
maturity that is small relatively  to the  mean reversion time for the volatility fluctuations. 
The corrections to the option prices and the
corresponding implied volatilities are identified there and the situation is
then qualitatively different from the one considered here in that the correction
to the price has a fractional behavior in time to maturity.
The characteristic  term structure exponent
 then reflects the roughness of the underlying  
volatility path, see Section \ref{sec:summ}  below.  
In \cite{sv2} we consider the case with stochastic volatility fluctuations
that are fast mean reverting  and that have a standard deviation of the same order as 
the mean, as we do here.
However, in \cite{sv2} the stochastic volatility
fluctuations have long-range correlation properties  so that the paths are smoother than those
associated with the Markov case.  Then the ``persistence''  of the volatility
fluctuations  leads to a fractional term structure and it also leads to a random component of
the price correction that is adapted to the filtration generated by  the price process
and  whose covariance structure can be  identified  in detail. 
 As we explain below
in our modeling the Hurst exponent $H$ parameterizes the smoothness of
the paths  with $H<1/2$ corresponding to short-range or rough paths
considered here 
and $H>1/2$ producing the long-range case. 
Indeed both
 regimes have been identified from the empirical perspective. We refer to  for instance
 \cite{gatheral1} for observations of rough volatility, 
 while  in  \cite{viens1} cases of volatility with  long-range  correlation properties are reported.
   Long-range volatility situations have    been reported 
for currencies in \cite{curr},  for commodities in \cite{energy} and for equity index in
\cite{mal}, while  analysis of electricity markets  data typically gives  $H<1/2$ as in
\cite{simonsen,rypdal,bennedsen}. 
  We  believe  that both the rough and the long-range cases   are important and can be observed
depending on the specific market and regime. 

  Taken together with the present paper the papers  \cite{sv1,sv2}
   give  a generalization
of the two-factor model of \cite{25,fouque11} to the case of 
processes with multiscale fluctuations, and for general smoothness
of the volatility factor, that has either short-range or long-range correlation properties.
Here we consider fractional  volatility in the context of option pricing, see
\cite{fH} and \cite{fH2} for applications to portfolio optimization.  
    
We also remark that,  albeit not treated in detail here,  model calibration
is an important issue for the practical use of the models.
An important aspect of the asymptotic 
results derived in our paper is   their use in calibration.
They provide a tool for robust calibration because they identify 
the parameters, the group market parameters,  that are important 
in affecting prices on contracts written on the  underlying.
Thus, the asymptotic results  can be used as a tool to avoid overfitting and noisy parameter 
estimates. Indeed,  the analysis of this paper shows that even at the level
of the correction the Hurst parameter does not affect the price.
Thus, for long dated contracts   the calibration  
scheme should not aim to identify the Hurst exponent.
 We also comment that the asymptotic results identify generic parameters,
   parameters that are common for different  contracts written on the same
 underlying.  The results presented here mean that the same  calibration
 scheme as that used in the Markovian case is appropriate.  
   Calibration  schemes for the Markovian case
 are considered  in for instance \cite{25,smile,fouque11}.
 The framework considered there is to calibrate parameters from 
 the implied volatility surface. We remark that estimates of volatility,
 spot volatility or proxies like VIX index  for instance,  also
 provide relevant information for calibration. Again here the asymptotic analysis
 provides an important tool because it identifies how basic aspects
 of the underlying affect the observables and how to connect information
 from for instance the  VIX index to pricing of financial contracts in the various
 asymptotic regimes.  Here our focus is on  fast mean  reverting rough volatility
 and how it affects pricing and thus implied volatility, but the tools presented can also be used
 to analyze issues associated with using historical data and for instance the VIX index
 for calibration.   
 
 One may ask why the Hurst exponent does not affect the implied volatility,
 as we will show below, even at the order of the correction.  The intuition
 for this is that the rough case with small Hurst exponent $H<1/2$ is mainly
 important on short time scales as it is the roughness, the rapid decay of correlations
 at the origin,  which primarily distinguishes the situation from the Markovian context.   
 The   correlations decay fast for large offsets  so that the correlation function
 for the rough volatility process  is integrable. 
 Thus, on time scales long relative to the mean reverting time of the 
rough volatility process, from a ``birds eye perspective'', 
 the roughness is not felt  and the process appears as a Markovian process.     
  This is in contrast to the long-range case with $H>1/2$ where the correlations 
 persist for a long time and the correlation function has heavy tails and is not integrable. 
 Then  the effect is felt also for long times. Indeed, we show in \cite{sv2} that
 with $H>1/2$,  and in the asymptotic framework considered there,  the Hurst
 exponent indeed has a strong impact on the form of the prices and hence
  the implied volatility. 
  In fact such a dichotomy    in terms of behavior and analytical approach
  depending on   $H\leq 1/2$ versus  $H>1/2$ can also 
  be observed in modeling of physical systems. See for instance the recent
  paper \cite{viens} which analyzes  the behavior of a Lyapunov coefficient  in the context of a driving fractional Brownian motion.

In Section \ref{sec:summ}  we summarize the form of the fractional term structure
exponent as it depends on the smoothness of volatility fluctuations,  
fluctuation magnitude, and the time  scale of mean reversion.
Otherwise the outline of the paper is as follows:
In Section \ref{sec:fou} we introduce the volatility factor in terms
of a fractional Ornstein--Uhlenbeck process and in Section 
\ref{sec:svm} the full stochastic volatility model. 
In Section \ref{sec:option} we present the main result
and its proof.  A number of technical lemmas that are  used in the proof
are presented in the appendices.   
 
\section{The Rapid Fractional Ornstein--Uhlenbeck Process}
\label{sec:fou}
We use a rapid fractional Ornstein--Uhlenbeck (fOU) process as the volatility factor and
describe here how this process can be represented in terms of a fractional Brownian motion.
Since fractional Brownian motion can be expressed in terms of ordinary Brownian motion
we also arrive at an expression for the rapid fOU process as a filtered version
of Brownian motion.

A fractional Brownian motion (fBM) is a zero-mean Gaussian process $(W^H_t)_{t\in \RR}$  
with the covariance
\begin{equation}
\label{eq:covfBM}
\EE[ W^H_t W^H_s ] = \frac{\sigma^2_H}{2} \big( |t|^{2H} + |s|^{2H} - |t-s|^{2H} \big) ,
\end{equation}
where $\sigma_H$ is a positive constant.

We use the following moving-average stochastic integral representation of the
fBM (see \cite{mandelbrot}):
\begin{equation}\label{mandelbrot}
W^H_t = \frac{1}{\Gamma(H+\frac{1}{2})} 
\int_{\RR} (t-s)_+^{H - \frac{1}{2}} -(-s)_+^{H - \frac{1}{2}} dW_s ,
\end{equation}
where $(W_t)_{t \in \RR}$ is a standard Brownian motion over $\RR$.
Then  indeed  $(W^H_t)_{t\in \RR}$ is a zero-mean Gaussian process with
the covariance (\ref{eq:covfBM}) and we have
\begin{eqnarray}
\nonumber
\sigma^2_H &=& \frac{1}{\Gamma(H+\frac{1}{2})^2} 
\Big[ \int_0^\infty \big( (1+s)^{H - \frac{1}{2}} -s^{H - \frac{1}{2}} \big)^ 2 ds +\frac{1}{2H}\Big] \\
&=&  \frac{1}{\Gamma(2H+1) \sin(\pi H)}.
\end{eqnarray}

We introduce the $\eps$-scaled fOU process as 
\begin{equation}\label{eq:fOU}
Z^\eps_t = \eps^{-H} \int_{-\infty}^t e^{-\frac{t-s}{\eps}} dW^H_s  = 
\eps^{-H}  W^H_t - \eps^{-1-H}  \int_{-\infty}^t e^{-\frac{t-s}{\eps}} W^H_s ds.
\end{equation}
Thus, the fOU process  is in fact a fractional Brownian motion with a restoring force towards zero.
It is a zero-mean, stationary Gaussian process,
with variance
\begin{equation}\label{eq:sou}
\EE [ (Z^\eps_t)^2 ] = \sigma^2_{{\rm ou}}, \mbox{ with } \sigma^2_{{\rm ou}} = \frac{1}{2} \Gamma(2H+1)  \sigma^2_H 
=\frac{1}{2\sin(\pi H)} ,
\end{equation}
that is independent of $\eps$,
and covariance:
\begin{eqnarray*}
\EE [ Z^\eps_t Z^\eps_{t+s}  ] &=&\sigma^2_{{\rm ou}} {\cal C}_Z\Big(\frac{s}{\eps}\Big) ,
\end{eqnarray*}
that is a function of $s/\eps$ only, with
\begin{eqnarray}
\nonumber
{\cal C}_Z(s) &=& 
 \frac{1}{\Gamma (2H+1)}
 \Big[ \frac{1}{2} \int_\RR e^{- |v|}| s+v|^{2H} dv - |s|^{2H}\Big]
\\
&=&
\frac{2 \sin (\pi H)}{\pi}
\int_0^\infty \cos (s x  ) \frac{x^{1-2H}}{1+x^2} dx .
\label{def:calCZ}
\end{eqnarray}
This shows that $\eps$ is the natural scale of variation of the fOU $Z^\eps_t$.
Note that the random process $Z^\eps_t$ is not  a martingale, nor  a Markov process.
For $H \in  (0,1/2)$ it possesses short-range correlation properties in the sense that its correlation
function is rough at zero:
\begin{equation}
\label{eq:corZG3}
{\cal C}_Z(s) =  1 - 
 \frac{1}{  \Gamma(2H+1)}  s^{2H}
+ o\big( s^{2H}\big) 
, \quad \quad s \ll 1,
\end{equation}
while it is integrable and it decays as $s^{2H-2}$ at infinity:
\begin{equation}
\label{eq:corZG3b}
{\cal C}_Z(s) =   
 \frac{1}{  \Gamma(2H-1)}  s^{2H-2}
+ o\big( s^{2H-2}\big) 
, \quad \quad s \gg 1.
\end{equation}

Using Eqs. (\ref{mandelbrot}) and (\ref{eq:fOU}) we arrive at
the moving-average integral representation of the scaled fOU as:
\begin{equation}
\label{eq:fOU2}
Z^\eps_t = \sigma_{{\rm ou}} \int_{-\infty}^t {\cal K}^\eps(t-s) dW_s,
\end{equation}
where 
\begin{equation}
\label{def:Keps}
{\cal K}^\eps(t) = \frac{1}{\sqrt{\eps}} {\cal K}\Big(\frac{t}{\eps}\Big),\quad \quad
{\cal K}(t) =\frac{1}{\sigma_{\rm ou}\Gamma(H+\frac{1}{2})} 
 \Big[ t^{H - \frac{1}{2}} - \int_0^t (t-s)^{H - \frac{1}{2}} e^{-s} ds \Big] .
\end{equation}
The main properties of the kernel ${\cal K}$ in 
our context are the following ones (valid for any $H \in (0,1/2)$):
\begin{enumerate}
\item[(i)]
${\cal K} \in L^2(0,\infty)$ with $\int_0^\infty {\cal K}^2(u) du = 1$
and ${\cal K} \in L^1(0,\infty)$.
\item[(ii)] For small times $t \ll 1$:
\begin{equation}
{\cal K} (t) = \frac{1}{\sigma_{\rm ou}\Gamma(H+\frac{1}{2})  } 
\Big( t^{H - \frac{1}{2}}  + O\big(  t^{H +\frac{1}{2}} \big) \Big) .
\end{equation}
\item[(iii)] For large times $t \gg 1$:
\begin{equation}
{\cal K} (t) = \frac{1}{\sigma_{\rm ou}\Gamma(H-\frac{1}{2})}  
\Big( t^{H - \frac{3}{2}}  +O\big(  t^{H - \frac{5}{2}} \big) \Big) .
\end{equation}
\end{enumerate}

{\bf Remark.}
The results presented in this paper can be generalized to any stochastic volatility model
of the form (\ref{def:stochmodel}) and (\ref{eq:fOU2}) provided 
${\cal K}^\eps(t) =   {\cal K} ({t}/{\eps}) / \sqrt{\eps}$ is such that the kernel ${\cal K}$
satisfies the properties (i)-(ii)-(iii) 
up to multiplicative constants.

\section{The Stochastic Volatility Model}
\label{sec:svm}%
The price of the risky asset follows the stochastic differential equation:
\begin{equation}
dX_t = \sigma_t^\eps X_t dW^*_t .
\end{equation}
The stochastic volatility is
\begin{equation}
\label{def:stochmodel}
\sigma_t^\eps = F(Z_t^\eps) ,
\end{equation}
where $Z_t^\eps$ is the scaled fOU with Hurst parameter $H \in (0,1/2)$ 
introduced in the previous section  which is  adapted to the Brownian motion $W_t$. Moreover,
 $W^*_t$ is a Brownian motion that is correlated to the stochastic volatility through
\begin{equation}\label{eq:corr} 
W^*_t = \rho W_t + \sqrt{1-\rho^2} B_t ,
\end{equation}
where the Brownian motion $B_t$ is independent of $W_t$.

The function $F$ is assumed to be one-to-one, positive-valued,
smooth, bounded and with bounded derivatives.
Accordingly, the filtration ${\cal F}_t $ generated
by $(B_t,W_t)$ is also the one generated by $X_t$.
Indeed, it is equivalent to the one generated by $(W^*_t,W_t)$, or $(W^*_t,Z^\eps_t)$.
Since $F$ is one-to-one, it is equivalent to the one generated by $(W^*_t,\sigma_t^\eps)$.
Since $F$ is positive-valued,
it is equivalent to the one generated by $(W^*_t,(\sigma_t^\eps)^2)$, or $X_t$. 


As we have discussed above, the volatility driving process  $Z_t^\eps$ has
short-range correlation properties. As we now show the volatility process
$\sigma_t^\eps$ inherits this property.

 \begin{lemma}
\label{prop:process}
We denote, for $j=1,2$:
\begin{equation}
\label{def:meanF}
\left< F^j \right> =  \int_{\RR}  F(\sigma_{{\rm ou}}z)^j p(z) dz,
\quad \quad 
\left< {F'}^j \right> =  \int_{\RR}  F'(\sigma_{{\rm ou}}z)^j p(z) dz,
\end{equation}
where $p(z)$ is the pdf of the standard normal distribution.

\begin{enumerate}
\item
The process $\sigma_t^\eps$ is a
stationary random process with mean $\EE[\sigma_t^\eps] = \left<F\right>$ and variance $
{\rm Var}(\sigma_t^\eps) = \left<F^2\right>-\left<F\right>^2$, independently of $\eps$.

\item
The covariance function of the process $\sigma_t^\eps$ is of the form
\begin{eqnarray}
\label{eq:corrY1}
{\rm Cov}\big( \sigma_t^\eps , \sigma_{t+s}^\eps \big) &=& \big( \left<F^2\right>-\left<F\right>^2\big) {\cal C}_\sigma \Big(\frac{s}{\eps}\Big),
\end{eqnarray}
where the correlation function ${\cal C}_\sigma$ satisfies ${\cal C}_\sigma(0)=1$ and 
\begin{eqnarray}
{\cal C}_\sigma(s) &=&
1 -  \frac{1   }{  \Gamma(2H+1)}
\frac{ \sigma_{{\rm ou}}^2\left< {F' }^2\right> }{ \left< F^2\right>-\left<F\right>^2}
s^{2H} +o\big( s^{2H}  \big)  ,
\quad \mbox{ for  } s \ll 1 ,
\label{eq:corrY12}\\
{\cal C}_\sigma(s) &=&
 \frac{1   }{  \Gamma(2H-1)}
\frac{ \sigma_{{\rm ou}}^2\left< F' \right>^2 }{ \left< F^2\right>-\left<F\right>^2}
s^{2H-2} +o\big( s^{2H-2}  \big)  ,
\quad \mbox{ for  } s \gg 1 .
\label{eq:corrY12b} 
\end{eqnarray}
\end{enumerate}
\end{lemma}
Consequently,  the process $\sigma_t^\eps$ has short-range correlation properties
and its covariance function is integrable.

\begin{proof}
  The fact that $\sigma_t^\eps$ is a stationary random process with mean
  $\left< F \right>$ is straightforward in view of the definition (\ref{def:stochmodel})
  of $\sigma_t^\eps$.  
 
For any $t,s$, the vector $\sigma_{{\rm ou}}^{-1} (Z^\eps_t ,Z^\eps_{t+s})$ is a Gaussian
random vector with mean $(0,0)$ and $2\times 2$ covariance matrix:
$$
{\bf C}^\eps =  \begin{pmatrix}
1 &{\cal C}_Z(s/\eps) \\
{\cal C}_Z(s/\eps) & 1 \end{pmatrix}.
$$
Therefore, denoting ${F}_c(z) = F(\sigma_{{\rm ou}}z) - \left< F \right>$,
the covariance function of the process $\sigma_t^\eps$ is
\begin{align*}
&{\rm Cov} (\sigma_t^\eps , \sigma_{t+s}^\eps) 
=  \EE \big[{F}_c(\sigma_{{\rm ou}}^{-1} Z_t^\eps)  {F}_c(\sigma_{{\rm ou}}^{-1} Z_{t+s}^\eps) \big]
\\
&=
\frac{1}{2 \pi  \sqrt{ \det {\bf C}^\eps}}
\iint_{\RR^2} {F}_c(z_1) {F}_c(z_2) 
\exp\Big( -\frac{1}{2}
\begin{pmatrix} z_1\\z_2\end{pmatrix}^T  
{{\bf C}^\eps}^{-1}
\begin{pmatrix} z_1\\z_2 \end{pmatrix} \Big) dz_1 dz_2 
\\
&= 
\Psi\Big( {\cal C}_Z \Big(\frac{s}{\eps}\Big)\Big) ,
\end{align*}
with
$$
\Psi(C) = 
\frac{1}{2 \pi \sqrt{1- C^2} }
\iint_{\RR^2}  {F}_c(z_1) {F}_c(z_2) \exp
\Big( - \frac{z_1^2+z_2^2 -2  C  z_1 z_2}{2 (1 -
C^2 )} \Big) dz_1 dz_2 \, .
$$
This shows that ${\rm Cov} (\sigma_t^\eps , \sigma_{t+s}^\eps) $ is a function of $s/\eps$ only.\\
The function $\Psi$ can be expanded in powers of $1-C$ for $C$ close to one:
\begin{eqnarray*}
\Psi(C) &=& 
\frac{1}{2 \pi }
\iint_{\RR^2} F_c\Big(z \frac{\sqrt{1+C}}{\sqrt{2}} +\zeta\frac{\sqrt{1-C}}{\sqrt{2}}\Big)F_c\Big(z\frac{\sqrt{1+C}}{\sqrt{2}}  -\zeta\frac{\sqrt{1-C}}{\sqrt{2}}\Big) \\
&&\times
\exp\Big( - \frac{z^2}{2} - \frac{\zeta^2}{2}\Big) dz d\zeta 
\\
&=&\frac{1}{\sqrt{2\pi}  }
\int_{\RR}  {F}_c(z)^2\exp
\Big( - \frac{z^2 }{2  }\Big) dz 
\\
&&+ (1-C) 
\frac{1}{\sqrt{2\pi} }
\int_{\RR}  {F}_c'(z)^2  \exp
\Big( - \frac{z^2 }{2 }\Big) dz  +\dessous{O}{C\to 1} ((1-C)^2), 
\end{eqnarray*}
which gives with  (\ref{eq:corZG3})  the form  (\ref{eq:corrY12}) 
of the 
correlation function for  $\sigma_t^\eps$.\\
The function $\Psi$ can be expanded in powers of $C$ for small $C$:  
\begin{eqnarray*}
\Psi(C) &=& 
\frac{1}{2\pi  }
\iint_{\RR^2}  {F}_c(z_1) {F}_c(z_2) \exp
\Big( - \frac{z_1^2+z_2^2 }{2  }\Big) dz_1 dz_2
\\
& &   -
\frac{C}{2\pi }
\iint_{\RR^2}  z_1 z_2 {F}_c(z_1) {F}_c(z_2) \exp
\Big( - \frac{z_1^2+z_2^2 }{2 }\Big) dz_1 dz_2 +
\dessous{O}{C\to 0} 
(C^2),  
\end{eqnarray*} 
which gives with (\ref{eq:corZG3b})  the form (\ref{eq:corrY12b}) 
of the correlation function for  $\sigma_t^\eps$.
\end{proof}

\section{The Option Price}
\label{sec:option}%
We aim at computing the option price defined as the martingale
\begin{equation}
M_t =\EE\big[ h(X_T) |{\cal F}_t \big]  ,
\end{equation}
where $h$ is a smooth payoff function and $t \leq T$. 
In fact weaker assumptions are possible for $h$,
as we only need to control the function $Q^{(0)}_t(x)$ defined below rather than $h$,
as is discussed in  \cite[Section 4]{sv1}.

We introduce the operator 
\begin{equation}
{\cal L}_{\rm BS} (\sigma) = \partial_t +\frac{1}{2} \sigma^2 x^2 \partial_x^2 ,
\end{equation}
that is, the standard Black--Scholes operator at zero interest rate and
constant volatility $\sigma$.   

We next exploit the fact that the price process is a
martingale to obtain an approximation, via constructing an explicit
function  $Q_t^\eps(x)$ so that $Q_T^\eps(x)=h(x)$ and so that $Q_t^\eps(X_t)$ 
is a martingale to first-order corrected  terms. Then, indeed $Q_t^\eps(X_t)$ 
gives the approximation  for the option price $M_t$ to this order.
   
The following proposition gives the first-order correction to the 
expression for  the martingale $M_t$ in the regime where $\eps$ is small.
\begin{proposition}\label{prop:main}
We have
\begin{equation}
\lim_{\eps \to 0}  \eps^{-1/2} \sup_{t\in [0,T]}   \EE\left[ | M_t  - Q_t^\eps(X_t) |^2 \right]^{1/2}  = 0  ,
\end{equation} 
where
\begin{equation}
\label{def:Qt}
Q_t^\eps(x) = Q_t^{(0)}(x) +\eps^{1/2}   \rho Q_t^{(1)}(x)   ,
\end{equation}
$Q_t^{(0)}(x) $ is deterministic and given by the Black--Scholes formula with constant volatility~$\overline{\sigma}$,
\begin{equation}
\label{eq:bs0}
{\cal L}_{\rm BS} (\overline\sigma) Q_t^{(0)}(x) =0,   \quad \quad Q_T^{(0)}(x) = h(x),
\end{equation}
with
\begin{equation} 
\label{def:barsigma}
\overline{\sigma}^2 = \left< F^2\right> = \int_\RR F(\sigma_{{\rm ou}} z )^2 p(z) dz,
\end{equation}
$p(z)$ is the pdf of the standard normal distribution,\\
$Q_t^{(1)}(x) $ is the deterministic correction
\begin{equation}
\label{def:Q1t}
Q_t^{(1)}(x) =(T-t)\overline{D}  \big( x \partial_x  (x^2 \partial_x^2 ) \big) Q_t^{(0)} (x)   ,
\end{equation}
 $\overline{D}$ is the coefficient defined by 
\begin{equation}
\overline{D} =   \sigma_{{\rm ou}} \int_0^\infty
\Big[\iint_{\RR^2} F(\sigma_{{\rm ou}} z )  (FF')(\sigma_{{\rm ou}} z') p_{{\cal C}_Z(s)} (z,z') d z dz'\Big]
 {\cal K}(s) ds,
\label{def:DtT}
\end{equation}
$p_C(z,z')$ is the pdf of the bivariate normal distribution with mean zero and covariance matrix $\begin{pmatrix}
1 & C\\ C & 1
\end{pmatrix}$,
and ${\cal C}_Z(s)$ is given by (\ref{def:calCZ}).
\end{proposition}

This proposition shows that the result is similar to the mixing (Markov) case addressed in \cite{fouque00,fouque11}.
In the fast-varying framework, the short-range correlation property of the stochastic volatility
is not visible to leading order nor  in the first correction. 
This is in contrast to the slowly-varying case addressed in \cite{sv1}
and this is the main result of this paper.

\subsection{The Case of  Riemann-Liouville Fractional Brownian Motion}

A common approach for modeling with fractional processes is to use
the Riemann-Liouville fractional Brownian motion defined by: 
\begin{equation}
\label{def:WH00}
W^{H,0}_t=  \frac{1}{\Gamma(H+\frac{1}{2})} \int_0^t (t-s)^{H-\frac{1}{2}}dW^0_s ,
\end{equation}
where $W^0$ is a standard Brownian motion.
From a mathematical perspective this is convenient as compared to the fractional Brownian
motion in  (\ref{mandelbrot}) because one does not have to deal with the integral for negative
times and the associated compensator.  However, from the modeling perspective
it has the disadvantage that the time zero plays a special role and the process
does not have stationary increments. 
By modeling the driving process as in (\ref{mandelbrot})  we obtain on the other hand a 
time homogeneous process.  
 In  \cite{fukasawa15}    a somewhat different approach to modeling with
 time-homogeneous fractional processes  is  used by using  a representation of fractional 
Brownian motion introduced by \cite{muravlev}.   
 In any case, we can use the representation in (\ref{def:WH00}) to define
a fOU process analogous to  (\ref{eq:fOU}) by  
 \begin{eqnarray}
\nonumber
Z^{\eps,0}_t &=& Z_0 e^{-t/\eps} +\eps^{-H} \int_0^t e^{-\frac{t-s}{\eps}} dW^{H,0}_s\\
&=& Z_0 e^{-t/\eps} +\eps^{-H} W^{H,0}_t - \eps^{-H-1} 
\int_0^t  e^{-\frac{t-s}{\eps}} W^{H,0}_s ds ,
\end{eqnarray}
where $Z_0$ is considered as a constant. In the  modeling context considered here,  from 
the point of view of the process covariance, the time epoch of  negative times is  in fact  quickly
forgotten so that  for any $t>0$, $s \geq 0$:
\ban
\lim_{\eps \to 0} {\rm Cov} \big( Z^{\eps,0}_t , Z^{\eps,0}_{t+\eps s}  \big) =
 \lim_{\eps \to 0}    {\rm Cov} \big( Z^{\eps}_t , Z^{\eps}_{t+\eps s}  \big) =\sigma^2_{{\rm ou}} {\cal C}_Z(s),
\ean
 and  in fact the covariances  only differ for  a time epoch of duration $\eps$ after
time zero.  The consequence is  that Proposition  \ref{prop:main}  
holds true when $Z^{\eps}_t$ is replaced by $Z^{\eps,0}_t$.
We discuss this in more detail in Appendix \ref{sec:appB}.
 
\subsection{Sketch  of Proof of Proposition \ref{prop:main}}

Our objective  is to  construct an approximation $Q_t^\eps(X_t)$ for the price.
Note that a natural first choice is to choose the approximation as
$Q_t^{(0)}(X_t)$, that is the Black--Scholes price
at the effective volatility.
In order to construct  a higher-order approximation we look for a correction,  
denoted by $\Delta Q_t^\eps(X_t)$, so that  
\ban
Q_t^{(0)}(X_t)+ \Delta Q_t^\eps(X_t) = \hbox{martingale} + \hbox{``small terms''} ,
\ean
and with $\Delta Q_T^\eps(X_T)=0$ so that the corrected approximation has the correct 
payoff. Indeed, the corrected approximation now differs from the exact price only
by the magnitude of the ``small terms'' because the exact price is a martingale.
Moreover, the volatility fluctuation 
process $(\sigma_s^\eps)^2 -\overline{\sigma}^2$ drives
the difference between the exact price and $Q_t^{(0)}(X_t)$.
To identify the form of the price correction  $\Delta Q_t^\eps(X_t)$ and to prove 
 the smallness of the  resulting (non-martingale)  error terms  the  
 introduction of  the martingale  defined in terms of the residual volatility fluctuations
is useful. That is why we introduce the process
\begin{equation}
\label{def:Kt0}
\psi_t^\eps = 
\EE \Big[  \frac{1}{2} \int_0^T \big( (\sigma_s^\eps)^2 -\overline{\sigma}^2 \big) ds \big| {\cal F}_t\Big] .
\end{equation}
This martingale is zero in the constant volatility case. 
It is important to understand its properties
and those of its covariation process with respect  to the underlying 
driving Brownian motion   in order to prove the accuracy of the approximation
and to control the error terms. 
These properties are given in terms of the original technical lemmas in Appendix \ref{sec:app}.
These lemmas could be useful also for the asymptotic analysis of other quantities
than those considered here, but defined in terms of 
underlyings modeled as in this paper.

\subsection{Proof of Proposition \ref{prop:main}}
For any smooth function $q_t(x)$, we have by It\^o's formula
\begin{eqnarray*}
dq_t(X_t) &=& \partial_t q_t(X_t) dt +  \big( x \partial_x\big) q_t (X_t) \sigma_t^\eps dW_t^*
+\frac{1}{2}  \big( x^2 \partial_x^2\big) q_t (X_t) (\sigma_t^\eps)^2 dt\\
&=&
{\cal L}_{\rm BS}(\sigma_t^\eps) q_t(X_t) dt +  \big( x \partial_x\big) q_t(X_t) \sigma_t^\eps dW_t^*  ,
\end{eqnarray*}
the last term being a martingale.
Here and below $ \big( x \partial_x\big) q_t(X_t)$ stands for $ x \partial_x  q_t(x)$ evaluated at ${x=X_t}$.
Therefore, by (\ref{eq:bs0}), we have
\begin{eqnarray}
dQ_t^{(0)}(X_t)  &=&
\frac{1}{2}\big( (\sigma_t^\eps)^2-\overline{\sigma}^2\big) \big( x^2 \partial_x^2 \big) Q_t^{(0)}(X_t) dt + dN^{(0)}_t ,
\end{eqnarray}
with $N_t^{(0)}$ a martingale:
$$
dN_t^{(0)}= \big( x \partial_x \big)  Q_t^{(0)}  (X_t) \sigma_t^\eps dW_t^* .
$$

Let $\phi_t^\eps$ be defined by
\begin{equation}
\label{def:phit}
\phi_t^\eps= \EE\Big[ \frac{1}{2} \int_t^T \big( (\sigma_s^\eps)^2 -\overline{\sigma}^2 \big) ds \big| {\cal F}_t\Big] .
\end{equation}
We have
$$
\phi_t^\eps = \psi_t^\eps - \frac{1}{2} \int_0^t \big( (\sigma_s^\eps)^2 -\overline{\sigma}^2 \big) ds ,
$$
where the martingale $\psi_t^\eps$ is defined by 
\begin{equation}
\label{def:Kt}
\psi_t^\eps = 
\EE \Big[  \frac{1}{2} \int_0^T \big( (\sigma_s^\eps)^2 -\overline{\sigma}^2 \big) ds \big| {\cal F}_t\Big] .
\end{equation}
We can write
$$
\frac{1}{2}\big( (\sigma_t^\eps)^2-\overline{\sigma}^2\big)
 \big( x^2 \partial_x^2 \big) Q_t^{(0)}(X_t) dt =
 \big( x^2 \partial_x^2 \big) Q_t^{(0)}(X_t) d\psi_t^\eps -
 \big( x^2 \partial_x^2 \big) Q_t^{(0)}(X_t) d\phi_t^\eps .
$$
By It\^o's formula:
\begin{eqnarray*}
d \big[ \phi_t^\eps \big( x^2 \partial_x^2 \big) Q_t^{(0)}(X_t)\big] &=&
\big( x^2 \partial_x^2 \big) Q_t^{(0)}(X_t) d\phi_t^\eps+
\big( x\partial_x\big(  x^2 \partial_x^2 \big)\big) Q_t^{(0)}(X_t) \sigma_t^\eps \phi_t^\eps dW_t^*
\\
&&+ {\cal L}_{\rm BS}(\sigma_t^\eps)  \big(  x^2 \partial_x^2 \big) Q_t^{(0)}(X_t)  \phi_t^\eps dt
\\
&&
+ \big( x\partial_x\big(  x^2 \partial_x^2 \big)\big) Q_t^{(0)}(X_t) \sigma_t^\eps d\left< \phi^\eps ,W^*\right>_t  .
\end{eqnarray*}
Since ${\cal L}_{\rm BS} (\sigma_t^\eps)= {\cal L}_{\rm BS}(\overline{\sigma}) + \frac{1}{2} \big( (\sigma_t^\eps)^2 - \overline{\sigma}^2\big)\big( x^2 \partial_x^2 \big)$
and ${\cal L}_{\rm BS}(\overline{\sigma})  \big(  x^2 \partial_x^2 \big) Q_t^{(0)}(x)=0$, this gives
\begin{eqnarray*}
d \big[ \phi_t^\eps \big( x^2 \partial_x^2 \big) Q_t^{(0)}(X_t)\big] &=&
-\frac{1}{2} \big( (\sigma_t^\eps)^2 - \overline{\sigma}^2\big)
\big( x^2 \partial_x^2 \big) Q_t^{(0)}(X_t) dt\\
&&+\frac{1}{2} \big( (\sigma_t^\eps)^2 - \overline{\sigma}^2\big)
\big( x^2\partial_x^2\big(  x^2 \partial_x^2 \big)\big)  Q_t^{(0)}(X_t) \phi_t^\eps dt\\
&&+ \big( x\partial_x\big(  x^2 \partial_x^2 \big) \big) Q_t^{(0)}(X_t) \sigma_t^\eps  d\left< \phi^\eps ,W^*\right>_t  
\\
&&  +
\big( x\partial_x\big(  x^2 \partial_x^2 \big)\big) Q_t^{(0)}(X_t) \sigma_t^\eps \phi_t^\eps dW_t^*+
\big( x^2 \partial_x^2 \big) Q_t^{(0)}(X_t) d\psi_t^\eps
.
\end{eqnarray*}
We have $\left< \phi^\eps ,W^*\right>_t = \left< \psi^\eps ,W^*\right>_t = \rho \left< \psi^\eps ,W\right>_t$ and therefore
\begin{eqnarray*}
d \big[( \phi_t^\eps \big( x^2 \partial_x^2 \big) Q_t^{(0)}(X_t)\big]
&=& -\frac{1}{2} \big( (\sigma_t^\eps)^2 - \overline{\sigma}^2\big)
\big( x^2 \partial_x^2 \big) Q_t^{(0)}(X_t) dt\\
&&+\frac{1}{2} \big( (\sigma_t^\eps)^2 - \overline{\sigma}^2\big)
\big( x^2\partial_x^2\big(  x^2 \partial_x^2 \big)\big)   Q_t^{(0)}(X_t) \phi_t^\eps dt
\\
&&+ \rho \big( x\partial_x\big(  x^2 \partial_x^2 \big) \big) Q_t^{(0)}(X_t)  \sigma_t^\eps    d\left< \psi^\eps ,W \right>_t   \\
&&+ dN^{(1)}_t ,
\end{eqnarray*}
where $N^{(1)}_t$ is a martingale,
$$
dN^{(1)}_t =\big( x\partial_x\big(  x^2 \partial_x^2 \big)\big) Q_t^{(0)}(X_t) \sigma_t^\eps \phi_t^\eps dW_t^*
+
 \big( x^2 \partial_x^2 \big) Q_t^{(0)}(X_t) d\psi_t^\eps .
$$
Therefore
\begin{eqnarray}
\nonumber
d \big[ Q_t^{(0)}(X_t)+ \phi_t^\eps \big( x^2 \partial_x^2 \big) Q_t^{(0)}(X_t)\big]
&=&  \frac{1}{2} 
\big( x^2\partial_x^2\big(  x^2 \partial_x^2 \big)\big) Q_t^{(0)}(X_t) \big( (\sigma_t^\eps)^2 - \overline{\sigma}^2\big)\phi_t^\eps dt
\\
\nonumber
&&+ \rho   \big( x\partial_x\big(  x^2 \partial_x^2 \big) \big) Q_t^{(0)}(X_t) 
\sigma_t^\eps  \vartheta^\eps_{t}
dt
   \\
&&+  dN^{(0)}_t +dN^{(1)}_t  .
\end{eqnarray} 
Here,  we have introduced the covariation increments 
\begin{eqnarray}
d \left< \psi^\eps, W\right>_t =  \vartheta^\eps_{t} dt  ,
\end{eqnarray}
defined in  Lemma \ref{lem:1}. 

The deterministic function $Q^{(1)}_t$ defined by (\ref{def:Q1t}) satisfies
$$
{\cal L}_{\rm BS}(\overline{\sigma}) Q^{(1)}_t(x) = - \overline{D} 
\big( x\partial_x ( x^2 \partial_x^2 )\big) Q^{(0)}_t(x)   , \quad \quad Q^{(1)}_T(x) = 0.
$$
Applying It\^o's formula
\begin{eqnarray*}
dQ_t^{(1)}(X_t)  
&=&
{\cal L}_{\rm BS}(\sigma_t^\eps) Q_t^{(1)}(X_t) dt +  \big( x \partial_x Q_t^{(1)}\big) (X_t) \sigma_t^\eps dW_t^* \\
&=& {\cal L}_{\rm BS}(\overline\sigma) Q_t^{(1)}(X_t) dt + 
\frac{1}{2} \big( (\sigma_t^\eps)^2 -\overline{\sigma}^2 \big) \big( x^2 \partial_x^2 \big) Q_t^{(1)}(X_t) dt  \\
&&
+
 \big( x \partial_x \big) Q_t^{(1)}(X_t) \sigma_t^\eps dW_t^*  \\
 &=& \frac{1}{2} \big( (\sigma_t^\eps)^2 -\overline{\sigma}^2 \big) \big( x^2 \partial_x^2 \big) Q_t^{(1)}(X_t) dt 
 - \big( x\partial_x  ( x^2 \partial_x^2) \big)Q^{(0)}_t (X_t)   \overline{D} dt \\
&& + dN^{(2)}_t ,
\end{eqnarray*}
where $N^{(2)}_t$ is a martingale,
$$
dN^{(2)}_t =  \big( x \partial_x \big)  Q_t^{(1)}  (X_t) \sigma_t^\eps dW_t^*   .
$$
Therefore
\begin{eqnarray}
\nonumber
&&
d \big[ Q_t^{(0)}(X_t)+ \phi_t^\eps \big( x^2 \partial_x^2 \big) Q_t^{(0)}(X_t)
+\eps^{1/2} \rho  Q_t^{(1)}(X_t)  
\big] \\
\nonumber
&&
 = \frac{1}{2} 
\big( x^2\partial_x^2(  x^2 \partial_x^2)\big) Q_t^{(0)}(X_t) \big( (\sigma_t^\eps)^2 - \overline{\sigma}^2\big)\phi_t^\eps dt \\
\nonumber
&& \quad 
+\frac{\eps^{1/2}}{2} \rho  \big( x^2 \partial_x^2 \big) Q_t^{(1)}(X_t) \big( (\sigma_t^\eps)^2 -\overline{\sigma}^2 \big) dt  \\
\nonumber
 &&
\quad 
+  \rho   \big( x\partial_x (  x^2 \partial_x^2 ) \big) Q_t^{(0)}(X_t)  \big( \sigma_t^\eps  {\vartheta}_{t}^{\eps} - \eps^{1/2} \overline{D}\big) dt \\
&&
\quad
+ d N^{(0)}_t + dN^{(1)}_t +\eps^{1/2} \rho dN^{(2)}_t .
\label{eq:proof1b}
\end{eqnarray}
We next  show that the first three terms of the right-hand side of (\ref{eq:proof1b}) are smaller than $\eps^{1/2}$.
We introduce for any $t \in [0,T]$:
\begin{eqnarray}
R^{(1)}_{t,T} &=& \int_t^T  \frac{1}{2} \big( x^2\partial_x^2(  x^2 \partial_x^2 )\big) Q_s^{(0)}(X_s)  \big( (\sigma_s^\eps)^2 - \overline{\sigma}^2\big) \phi_s^\eps
 ds , \\
R^{(2)}_{t,T} &=& \int_t^T \frac{\eps^{1/2}}{2} \rho  \big( x^2 \partial_x^2 \big) Q_s^{(1)}(X_s)  \big( (\sigma_s^\eps)^2 -\overline{\sigma}^2 \big)ds , \\
R^{(3)}_{t,T} &=& \int_t^T  \rho   \big( x\partial_x(  x^2 \partial_x^2 ) \big) Q_s^{(0)}(X_s) 
 \big( \vartheta_{s}^\eps   \sigma_s^\eps-\eps^{1/2}  \overline{D} )ds  .
\end{eqnarray}

We will show that,  for $j=1,2,3$,
\begin{equation}
\label{eq:estimeRj}
\displaystyle \lim_{\eps \to 0} \eps^{-1/2} \sup_{t \in [0,T]} \EE \big[ (R^{(j)}_{t,T})^2 \big]^{1/2} =0 .
\end{equation}

{\it Step 1: Proof of (\ref{eq:estimeRj}) for $j=1$.}\\
Since $Q^{(0)}$ is smooth and bounded and $F$ is bounded, there exists $C$ such that
$$
 \sup_{t \in [0,T]}  \EE \big[ (R^{(1)}_{t,T})^2 \big] \leq C T \int_0^T \EE \big[ (\phi_s^\eps)^2 \big] ds .
$$
By Lemma \ref{lem:3} we get the desired result.\\

{\it Step 2: Proof of (\ref{eq:estimeRj}) for $j=2$.}\\
We denote 
$$
Y^{(2)}_s =   \rho  \big( x^2 \partial_x^2 \big) Q_s^{(1)}(X_s)  
$$
and
\begin{equation}
\label{def:kappaeps}
\kappa_t^\eps = \frac{\eps^{1/2}}{2}  \int_0^t \big( (\sigma_s^\eps)^2 -\overline{\sigma}^2 \big)ds  ,
\end{equation}
so that 
$$
R^{(2)}_{t,T} = \int_t^T Y^{(2)}_s \frac{d\kappa_s^\eps }{ds} ds .
$$
Note that $Y^{(2)}_s$ is a bounded semimartingale with bounded quadratic variations.
Let $N$ be a positive integer. We denote $t_k=t+(T-t)k/N$. We then have
\begin{eqnarray*}
R^{(2)}_{t,T} &=& \sum_{k=0}^{N-1} \int_{t_k}^{t_{k+1}}Y^{(2)}_s \frac{d\kappa_s^\eps }{ds} ds =R^{(2,a)}_{t,T} +R^{(2,b)}_{t,T}  ,\\
R^{(2,a)}_{t,T} &=&\sum_{k=0}^{N-1} \int_{t_k}^{t_{k+1}}Y^{(2)}_{t_k} \frac{d\kappa_s^\eps }{ds} ds = \sum_{k=0}^{N-1} Y^{(2)}_{t_k} \big( \kappa_{t_{k+1}}^\eps -
 \kappa_{t_{k}}^\eps \big) , \\
 R^{(2,b)}_{t,T} &=&\sum_{k=0}^{N-1} \int_{t_k}^{t_{k+1}}\big( Y^{(2)}_{s}-Y^{(2)}_{t_k}\big)  \frac{d\kappa_s^\eps }{ds} ds  .
\end{eqnarray*}
Then,  on the one hand
\begin{eqnarray*}
\EE \big[ (R^{(2,a)}_{t,T})^2 \big]^{1/2} &\leq & 2 \sum_{k=0}^{N} \|Y^{(2)}\|_\infty 
\EE[ (\kappa_{t_{k}}^\eps)^2]^{1/2}  \\
&\leq&  2  (N+1) \|Y^{(2)}\|_\infty \sup_{s\in [0,T]}
\EE[ (\kappa_{s}^\eps)^2]^{1/2}  ,
\end{eqnarray*}
so that, by Lemma \ref{lem:6},
$$
 \lim_{\eps \to 0} \eps^{-1/2} \sup_{t \in [0,T]} \EE \big[ (R^{(2,a)}_{t,T})^2 \big]^{1/2} =0 .
$$
On the other hand 
\begin{eqnarray*}
\EE \big[ (R^{(2,b)}_{t,T})^2 \big]^{1/2} &\leq & \eps^{1/2} \|F\|_\infty^2 \sum_{k=0}^{N-1} \int_{t_k}^{t_{k+1}} \EE[ \big( Y^{(2)}_{s}-Y^{(2)}_{t_k}\big)^2]^{1/2}ds \\
&\leq & K \eps^{1/2} \sum_{k=0}^{N-1} \int_{t_k}^{t_{k+1}} (s-t_k)^{1/2} ds  = \frac{2K T^{3/2}\eps^{1/2}}{3\sqrt{N}}   .
\end{eqnarray*}
Therefore, we get
\begin{eqnarray*}
 \limsup_{\eps \to 0} \eps^{-1/2} \sup_{t \in [0,T]} \EE \big[ (R^{(2)}_{t,T})^2 \big]^{1/2} &\leq & \limsup_{\eps \to 0} \eps^{-1/2} \sup_{t \in [0,T]} \EE \big[ (R^{(2,b)}_{t,T})^2 \big]^{1/2} \\ 
 &\leq & \frac{2KT^{3/2}}{3\sqrt{N}} .
\end{eqnarray*}
Since this is true for any $N$, we get the desired result.\\

{\it Step 3: Proof of (\ref{eq:estimeRj}) for $j=3$.}\\
We repeat the same arguments as in the previous step. It remains to show that
$$
\widetilde{\kappa}^\eps_t =  \int_0^t \big( \vartheta^\eps_s \sigma_s^\eps - \eps^{1/2} \overline{D} \big)ds  
$$
satisfies
$$
 \lim_{\eps \to 0} \eps^{-1/2} \sup_{t \in [0,T]} \EE \big[ (\widetilde{\kappa}^\eps_t)^2\big]^{1/2}  =0.
$$
Since $(a+b)^2\leq 2a^2 + 2b^2$,  we have
$$
\EE \big[ (\widetilde{\kappa}^\eps_t)^2\big] \leq
2 \int_0^t ds \int_0^t ds' {\rm Cov} 
\big( \vartheta^\eps_s \sigma_s^\eps ,\vartheta^\eps_{s'} \sigma_{s'}^\eps \big)
+2 \Big( \int_0^t \big(\EE [  \vartheta^\eps_s \sigma_s^\eps ] - \eps^{1/2} \overline{D}  \big) ds \Big)^2  .
$$
By Lemma \ref{lem:2}, items 1 and 3, and dominated convergence theorem, the first term of the right-hand side 
is $o(\eps)$ uniformly in $t\in [0,T]$.
By Lemma \ref{lem:2}, item 2, the second term of the right-hand side  is $o(\eps)$ uniformly in $t\in [0,T]$.
This gives the desired result.

We can now complete the proof of Proposition \ref{prop:main}.
We introduce the approximation: 
$$
\widetilde{Q}_t^\eps(x) = 
Q_t^{(0)}(x)+ \phi_t^\eps \big( x^2 \partial_x^2 \big) Q_t^{(0)}(x)
+\eps^{1/2} \rho Q_t^{(1)}(x).
$$
We then have 
$$
\widetilde{Q}_T^\eps(x) = h(x),
$$
because $Q_T^{(0)}(x)=h(x)$, $\phi_T^\eps=0$, and $Q^{(1)}_T(x)=0$.
Let us denote
\begin{eqnarray}
R_{t,T} &=&R^{(1)}_{t,T}+
R^{(2)}_{t,T}+
R^{(3)}_{t,T} ,\\
N_t &=&  \int_0^t d N^{(0)}_s + dN^{(1)}_s +\eps^{1/2} \rho dN^{(2)}_s  .
\end{eqnarray}
By (\ref{eq:proof1b}) we have
$$
\widetilde{Q}_T^\eps(X_T) - \widetilde{Q}_t^\eps(X_t) = R_{t,T} + N_T-  N_t.
$$
Therefore
\begin{eqnarray}
\nonumber
M_t &=& \EE \big[ h(X_T) |{\cal F}_t \big] = 
\EE \big[ \widetilde{Q}_T^\eps(X_T) |{\cal F}_t \big] \\
\nonumber
&=& \widetilde{Q}_t^\eps(X_t) +\EE \big[ R_{t,T} |{\cal F}_t \big]+
\EE \big[ N_T-N_t |{\cal F}_t \big] \\
&=&
\widetilde{Q}_t^\eps(X_t) +\EE \big[ R_{t,T} |{\cal F}_t \big]
 ,
\end{eqnarray}
which gives the desired result since $\EE \big[ R_{t,T}  |{\cal F}_t \big]$ and $\phi^\eps_t$ are 
{uniformly} of order $o(\eps^{1/2})$ in $L^2$
 (see (\ref{eq:estimeRj}) for $R_{t,T}$ and see Lemma \ref{lem:3} for $\phi^\eps_t$).

\section{The Implied Volatility}
\label{sec:implied}
We now compute and discuss the implied  volatility  associated with the price
approximation given in Proposition \ref{prop:main}.
This implied volatility  is the volatility that when used in the constant 
volatility Black--Scholes pricing formula gives the same  price as the approximation,
to  the order of the approximation.  
The implied volatility in the context of the European option 
introduced in the previous section is then given by
\begin{equation}
\label{eq:iv1}
I_t = \bar{\sigma} 
+ \eps^{1/2}  \rho \overline{D}\Big[ \frac{1}{2\bar{\sigma} } + \frac{\log(K/X_t)}{\bar{\sigma}^3 (T-t)} \Big]
+o( \eps^{1/2})  .
\end{equation}

The expression (\ref{eq:iv1}) is in agreement with the one obtained in \cite[Eq. (5.55)]{fouque00} with a stochastic
volatility that is an ordinary Ornstein--Uhlenbeck process, that is, 
a Markovian process with correlations decaying exponentially fast.
See for instance \cite{25,smile,fouque11} and references therein for data calibration examples.  
 
\section{A Brief Review on Fractional Stochastic Volatility Asymptotics}
\label{sec:summ}

This paper together with \cite{sv1,sv2} discuss 
different fractional stochastic volatility models with $H<1/2$ or $H>1/2$.
We summarize here some main aspects. 

\subsection{Characteristic Term Structure Exponent}

We write the implied volatility associated with a 
European Call Option for strike $K$, maturity $T$, current time $t$,
and current value for the underlying $X_t$ 
(as  in Eq.~(\ref{eq:iv1}))  in the general form:
 \begin{eqnarray}\label{eq:iv3}
I_t &=& \sigma_{t,T}
+   \Delta \sigma
\left[  \left( \frac{\tau}{\bar\tau} \right)^{\zeta(H)} 
+  \left( \frac{\tau}{\bar\tau} \right)^{\zeta(H)-1} {\log\Big(\frac{K}{X_t}\Big)} \right]
  , \\
 \sigma_{t,T}  &=&    \EE\Big[ \frac{1}{T-t} \int_t^T (\sigma_s)^2 ds \big| {\cal F}_t \Big]^{1/2} 
 =  \bar{\sigma} +   \tilde{\sigma}_{t,T} ,
 \label{eq:iv4}
\end{eqnarray}
where   $\sigma_s$ is the volatility path, $\bar{\tau}$ is the characteristic  diffusion time defined by
\begin{equation}
\label{def:bartau}
 \bar\tau = \frac{2}{{\overline\sigma}^2}   ,
\end{equation}
and $\tau=T-t$ is the  time to maturity. 
Note that in the regimes that we consider we assume that the  
time to maturity is of the same order as the characteristic diffusion time.  
We refer to 
$\zeta(H)$ as the characteristic term structure exponent.
 Note that $\tilde{\sigma}_{t,T}$ is  the price path predicted 
 volatility correction relative to the time horizon and is a stochastic
 process adapted to the filtration generated by the 
 underlying price process. This is a correction term that reflects
 the multiscale nature of the volatility fluctuations and the
 memory  aspect of this process. The second correction
 term in Eq. (\ref{eq:iv3}) involving $\Delta \sigma$ is a skewness correction  and 
 vanish in the case $\rho=0$. As mentioned above it is natural to let  the
 characteristic  diffusion time be the reference time scale, if we denote the 
 mean reversion  time of the volatility fluctuations  by
 $\tau_{\rm mr}$ then we have considered  two main  multiscale asymptotic 
 regimes in the context of the characteristic term structure exponent:
\begin{itemize}
 \item  Slow mean reverting     
  volatility fluctuations, $\tau_{\rm mr} \gg   \bar\tau$,  (see \cite{sv1}).
 In this case:
 \begin{equation}
\label{eq:zetaslow}
 \zeta(H)   =   H+\frac{1}{2}   .
 \end{equation}
\item Fast mean reverting  volatility fluctuations, 
$\tau_{\rm mr} \ll   \bar\tau$, (see \cite{sv2} for $H \in (1/2,1)$ and this paper
for $H\in (0,1/2)$).
 In this case:
 \begin{equation}
 \label{eq:zetafast}
 \zeta(H)   =   \max \Big(   H-\frac{1}{2} , 0\Big) .
 \end{equation}
 \end{itemize}
 Thus,  we see that in the case of fast mean reversion 
 leading to a  singular perturbation   expansion 
 we have a fractional characteristic term structure exponent
 only in the case $H>1/2$ when  we have long-range correlation properties.
 While in the slow mean reversion  case leading to a regular perturbation   expansion 
 we have a  fractional characteristic term structure exponent for all values of
 the Hurst exponent $H$.
 
 In \cite{sv1} we considered also the case of small volatility fluctuations  whose 
mean reversion time  is of the same order as  the characteristic diffusion time. 
 This leads 
 to an asymptotic regime where the  characteristic term structure exponent 
 is replaced by a more general  characteristic term structure factor   
 of the form (assuming a fOU volatility factor):
 \begin{eqnarray*}
     \Big( \frac{\tau}{\bar\tau} \Big)^{\zeta(H)} &  \rightarrow   &  
     {\cal A} \Big(\frac{\tau}{\bar\tau}, \frac{\tau}{\tau_{\rm mr}}\Big) =
   \Big( \frac{\tau}{\bar\tau} \Big)^{H+1/2} 
\Big\{ 1-\int_0^{\tau/\tau_{\rm mr}} e^{-v}
 \Big( 1-\frac{v}{\tau/\tau_{\rm mr}} \Big)^{H+\frac{3}{2}}  
dv \Big\}  .
 \end{eqnarray*}
 We then have in a subsequent limit of either  slow 
 ($\tau_{\rm mr}  \gg \tau$) or fast ($\tau_{\rm mr} \ll \tau $) mean reversion:
 \begin{equation}
 \label{eq:Apred}
  {\cal A} \Big(\frac{\tau}{\bar\tau}, \frac{\tau}{\tau_{\rm mr}}\Big) 
 \propto  \left\{
\begin{array}{ll}
   \left( \frac{\tau}{\bar\tau} \right)^{H+1/2}
   &   ~\hbox{for}~ \tau \ll \tau_{\rm mr}  ,  \\
   \left( \frac{\tau}{\bar\tau} \right)^{H-1/2}
   &   ~\hbox{for}~ \tau \gg  \tau_{\rm mr}   , 
\end{array}
\right.
\end{equation}
where we have a fractional term structure for all values of $H$.
It follows  that the characteristic term structure exponent is consistent
 with the result (\ref{eq:zetaslow}) obtained in the slow mean reverting limit.
 It is also consistent with the result (\ref{eq:zetafast}) obtained in the fast mean reverting limit,
 but only in the case $H \in (1/2,1)$.
 There is no contradiction in the case $H\in (0,1/2)$
 because there is no fundamental reason that would justify that 
 the limits ``small amplitude" and ``fast mean reversion" are exchangeable.
 This means that the prediction (\ref{eq:Apred}) for $H\in (0,1/2)$ in the limit ``small amplitude"
and then   ``fast mean reversion" does not capture the leading-order contribution of the limit  ``fast mean reversion"
 that is  independent of time to maturity,  but is negligible for small-amplitude volatility fluctuations.
 Note that when the standard deviation of the volatility 
 fluctuations is of the same order as the mean volatility and the time 
 to maturity is of the same order as the mean reversion time,
then the implied volatility reflects the particular structure of the model,
 see for instance the analysis of the Heston (\cite{heston})  model in   \cite{alos2}.
 Note also that with a model for how the implied volatility depends on  the Hurst exponent 
 we can actually estimate the Hurst exponent based on recordings of the implied volatility. 
An example  with estimation of the Hurst exponent  based on a spot  volatility
proxy deriving from implied volatility  is in \cite{livieri} and yields a rough volatility regime. 
 A calibration example  
for $H$  using VIX futures  is in \cite{jacquier} and yields  again a rough
volatility regime. 
    
 \subsection{Flapping of the Implied Surface}
 
Regarding the price path predicted 
 volatility correction $\tilde{\sigma}_{t,T}$  which depends on the price history we have the
 following picture in  the regime of fast  mean reversion:
 \begin{itemize}
\item Rough volatility fluctuations, $H<1/2$: 
 \begin{eqnarray}
  \tilde{\sigma}_{t,T}   &=&  o(\Delta \sigma)  .
  \end{eqnarray}
  \item  Smooth volatility fluctuations, $H>1/2$: 
 \begin{eqnarray}
  \tilde{\sigma}_{t,T}   &=&  O(\Delta \sigma)  .
 \end{eqnarray}
 \end{itemize}
In the case of smooth volatility fluctuations we discuss in detail 
in \cite{sv2} the statistical structure of the ``t-T'' process  $\tilde{\sigma}_{t,T}$. 
 We remark
that indeed in the  scaling  addressed in this paper we have $\Delta \sigma/\bar\sigma \ll 1$.  
In the regime of slow  mean reversion as discussed in \cite{sv1} we
have that  $\tilde{\sigma}_{t,T}   =  O(\bar\sigma)$ since then  the current level
of volatility  plays a central role.


\section{Conclusion}

We have considered rough fractional stochastic volatility models. Such modeling
is motivated by a number of recent empirical findings that the volatility is 
not well modeled by a Markov  process with exponentially decaying correlations and certainly
not by a constant. Rather  it should be modeled as a stochastic process with correlations that are rapidly decaying at the origin, 
qualitatively faster than the  decay that can be associated with a Markov process. 
 In general such models are challenging to use since  the volatility factor is not a
Markov process nor a martingale so we do not have a pricing partial differential equation.  
However, here we consider the situation where the volatility is fast mean reverting in the sense 
that its  mean reversion time is short relative to the characteristic diffusion time of the price process. 
In this regime the pricing problem and associated implied volatility 
surface can be reduced to a parametric form corresponding to that of the Markovian case.
An important aspect of our modeling is that we model the rough  stochastic volatility as being 
a stationary process. 
Many if not most papers  on this subject have hitherto used a non-stationary framework where the 
``time zero'' plays a special role.  In the case of processes with memory  of the past, rather than being Markov,
we consider this aspect to be crucial from the modeling viewpoint.
Indeed,  in the general case the history (in principle observable from the underlying price path)
impacts the implied volatility. 
However,  in the regime of fast mean reversion
the impact of the price  history becomes lower order relative to the leading correction
associated with the  stochastic volatility which is explicit and which is identified in this paper.    

It is important to note that this picture in fact breaks down in the case of long-range
stochastic volatility when the volatility factor paths are smoother than in the
Markovian case and when their correlations decay slower than in the Markovian case.
It also breaks down in the asymptotic context when the mean reversion time is of the same order as
the characteristic diffusion time of the price process, 
but volatility fluctuations have small standard deviation compared to the mean. 
In these cases both with short- (rough volatility) and long-range
correlation properties the structure of the model for the price correction and the implied volatility changes 
and leads to a picture with a fractional term structure. These results are  derived in 
\cite{sv1,sv2} and summarized in Section \ref{sec:summ} herein.    

These observations  then serve to partly explain why parameterizations 
for the implied surface deriving from a Markovian modeling 
have been successful in capturing
the implied volatility surface despite empirical observations that refute the Markovian
framework.  
Finally, this analytic result confirms the results of the recent paper \cite{funahashi} when
the price associated with a rough stochastic volatility model was computed numerically.    
  
\section*{Acknowledgements}
  This research  has been  supported in part by 
 Centre Cournot, Fondation Cournot, and 
 Universit\'e Paris Saclay (chaire D'Alembert).

\appendix
 
\section{Technical Lemmas}\label{sec:app}
We denote
\begin{equation}
\label{def:G}
G(z) = \frac{1}{2} \big( F(z)^2 - \overline{\sigma}^2\big) .
\end{equation}
The martingale $\psi^\eps_t$ defined by (\ref{def:Kt}) has the form
\begin{equation}
\psi_t^\eps = 
\EE \Big[   \int_0^T G(Z_s^\eps)  ds \big| {\cal F}_t\Big] .
\end{equation}

\begin{lemma}
\label{lem:1}%
$(\psi_t^\eps)_{t\in [0,T]}$ is a square-integrable martingale and
\begin{equation}
\label{def:varthetaeps}
d \left< \psi^\eps, W\right>_t =  \vartheta^\eps_{t} dt ,
\quad \quad \vartheta^\eps_{t} = \sigma_{{\rm ou}} \int_t^T \EE \big[ G'(Z_s^\eps)|{\cal F}_t \big]{\cal K}^\eps(s-t) ds   .
\end{equation}
\end{lemma}
\begin{proof}
See Lemma B.1 in \cite{sv2}.
\end{proof}

The important properties of the random process $\vartheta^\eps_{t}$ are stated in the following lemma.
\begin{lemma}
\label{lem:2}
\begin{enumerate}
\item The exists a constant $K_T$ such that,
for any $t \in [0,T]$, we have almost surely
\begin{equation}
\big| \sigma_t^\eps \vartheta^\eps_{t}  \big| \leq K_T \eps^{1/2}  .
\end{equation}
\item
For any $t \in [0,T]$, we have 
\begin{equation}
\EE [ \sigma_t^\eps \vartheta^\eps_{t} ] = \eps^{1/2}  \overline{D} + \widetilde{D}^\eps_{t},
\end{equation}
where $ \overline{D}$ is the deterministic constant (\ref{def:DtT})
and $\widetilde{D}^\eps_{t}$ is smaller than $\eps^{1/2}$:
\begin{equation} 
\sup_{\eps \in (0,1]} \sup_{t \in [0,T]}  \eps^{-1/2} \big| \widetilde{D}^\eps_{t}\big| < \infty,
\end{equation}
and
\begin{equation}
\forall t \in [0,T),\quad
\lim_{\eps \to 0} \eps^{-1/2}   \big| \widetilde{D}^\eps_{t}\big| =0.
\end{equation}
\item 
For any $0\leq t < t' < T$, we have
\begin{equation} 
 \lim_{\eps \to 0}  \eps^{-1} \big| {\rm Cov} 
\big( \sigma_t^\eps \vartheta^\eps_{t},\sigma_{t'}^\eps \vartheta^\eps_{t'}\big) \big| =0.
\end{equation}
\end{enumerate}
\end{lemma}

\begin{proof}
Using the expression (\ref{def:varthetaeps}) of $\vartheta^\eps_t$:
$$
\big| \vartheta^\eps_t \sigma^\eps_t \big| \leq \sigma_{\rm ou} \|F\|_\infty \|G'\|_\infty \int_0^\infty |{\cal K}^\eps(s)| ds
$$
The proof of the first item follows from
the fact that ${\cal K}^\eps(t)={\cal K}(t/\eps)/\sqrt{\eps}$, ${\cal K} \in L^1(0,\infty)$.

The expectation of $  \sigma_t^\eps\vartheta^\eps_{t}  $ is equal to
\begin{eqnarray*}
 \EE\big[ \sigma_t^\eps\vartheta^\eps_{t}  \big ] 
&=& \sigma_{{\rm ou}}   \int_t^{T} \EE\big[ F(Z^\eps_t) G'(Z_s^\eps) \big]  {\cal K}^\eps(s - t) ds \\
&=& \sigma_{{\rm ou}} \eps^{1/2} \int_0^{(T-t)/\eps} \EE\big[ F(Z_0^\eps) G'(Z_{\eps s}^\eps) \big] {\cal K}(s) ds \\
&=&  \sigma_{{\rm ou}} \eps^{1/2} \int_0^{(T-t)/\eps}
\Big[\iint_{\RR^2} F(\sigma_{{\rm ou}} z )  G'(\sigma_{{\rm ou}} z') p_{{\cal C}_Z(s)} (z,z') d z dz'\Big]
 {\cal K}(s) ds,
\end{eqnarray*}
with  $p_C$   defined in Proposition \ref{prop:main}.

Therefore the difference
$$
 \EE\big[ \sigma_t^\eps \vartheta^\eps_{t}  \big ]  - \eps^{1/2} \overline{D}
=
\sigma_{{\rm ou}}  \eps^{1/2} \int_{(T-t)/\eps}^\infty \EE\big[ F(Z_0^\eps) G'(Z_{\eps s}^\eps)\big] {\cal K}(s) ds
$$
can be bounded by
\begin{equation}
\big| \EE\big[  \sigma_t^\eps\vartheta^\eps_{t}  \big ]   -  \eps^{1/2} \overline{D} \big| 
\leq 
\|F\|_\infty \|G'\|_\infty \sigma_{{\rm ou}}  \eps^{1/2} \int_{(T-t)/\eps}^\infty |{\cal K}(s)| ds,
\end{equation}
which gives the second item since ${\cal K} \in L^1(0,\infty)$.

Let us consider $0 \leq t \leq t'\leq T$. We have
\begin{eqnarray*}
 \EE\big[ \sigma_t^\eps\vartheta^\eps_{t}\sigma_{t'}^\eps \vartheta^\eps_{t'}  \big] 
&=& \sigma_{{\rm ou}}^2 \int_t^{T} ds   {\cal K}^\eps(s-t)  \int_{t'}^T ds'  {\cal K}^\eps(s' -t')  \\
&&\times
\EE\Big[ 
\EE\big[ F(Z^\eps_t) G'(Z_s^\eps)|{\cal F}_t \big]  \EE\big[ F(Z^\eps_{t'}) G'(Z_{s'}^\eps)|{\cal F}_{t'} \big] \Big]  ,
\end{eqnarray*}
so we can write
\begin{eqnarray*}
{\rm Cov} \big( \sigma_t^\eps\vartheta^\eps_{t} , \sigma_{t'}^\eps \vartheta^\eps_{t'}  \big)
&=& \sigma_{{\rm ou}}^2 \int_t^{T} ds   {\cal K}^\eps(s-t)  \int_{t'}^T ds'  {\cal K}^\eps(s' - t')  \\
&&\times
\Big(
\EE\Big[ 
\EE\big[ F(Z^\eps_t) G'(Z_s^\eps)|{\cal F}_t \big]  \EE\big[ F(Z^\eps_{t'}) G'(Z_{s'}^\eps)|{\cal F}_{t} \big] \Big]  \\
&&\quad - 
\EE\Big[ 
\EE\big[ F(Z^\eps_t) G'(Z_s^\eps)|{\cal F}_t \big]  \EE\big[ F(Z^\eps_{t'}) G'(Z_{s'}^\eps)  \big] \Big] \Big)  ,
\end{eqnarray*}
and therefore
\begin{eqnarray*}
\big| {\rm Cov} \big( \sigma_t^\eps\vartheta^\eps_{t} , \sigma_{t'}^\eps \vartheta^\eps_{t'}  \big) \big|
&\leq & \sigma_{{\rm ou}}^2 \|F\|_\infty \|G'\|_\infty \int_t^{T} ds   |{\cal K}^\eps(s-t)| \int_{t'}^T ds'  |{\cal K}^\eps(s'-t')| \\
&&\times
\EE\Big[ \big( \EE\big[ F(Z^\eps_{t'}) G'(Z_{s'}^\eps)|{\cal F}_{t} \big] -   \EE\big[ F(Z^\eps_{t'}) G'(Z_{s'}^\eps) \big] \big)^2 \Big]^{1/2} .
\end{eqnarray*}
We can write for any $\tau>t$:
\begin{eqnarray*}
&&Z^\eps_{\tau} = A^\eps_{t\tau} +  B^\eps_{t\tau} ,\quad  A^\eps_{t\tau} = \sigma_{\rm ou} \int_{-\infty}^t {\cal K}^\eps(\tau-u)dW_u, \quad 
 B^\eps_{t\tau} =  \sigma_{\rm ou}\int_t^{\tau} {\cal K}^\eps(\tau-u)dW_u ,
\end{eqnarray*}
where $A^\eps_{t\tau}$ is ${\cal F}_t$ adapted while $B^\eps_{t\tau}$ is independent from ${\cal F}_t$.
Therefore ($s'\geq t'\geq t$)
\begin{eqnarray*}
&&\EE\Big[ \big( \EE\big[ F(Z^\eps_{t'}) G'(Z_{s'}^\eps)|{\cal F}_{t} \big] -   \EE\big[ F(Z^\eps_{t'}) G'(Z_{s'}^\eps)  \big]\big)^2 \Big]\\
&&= \EE\Big[   \EE\big[ F(Z^\eps_{t'}) G'(Z_{s'}^\eps)|{\cal F}_{t} \big]^2 \Big] -   \EE\big[ F(Z^\eps_{t'}) G'(Z_{s'}^\eps)  \big]^2 \\
&&
=
\EE\Big[
F( A^\eps_{tt'} +  B^\eps_{tt'}) G'( A^\eps_{ts'} +  B^\eps_{ts'})
F( A^\eps_{tt'} +  \tilde{B}^\eps_{tt'}) G'( A^\eps_{ts'} +  \tilde{B}^\eps_{ts'})
\\
&& \quad -
F( A^\eps_{tt'} +  B^\eps_{tt'}) G'( A^\eps_{ts'} +  B^\eps_{ts'})
F( \tilde{A}^\eps_{tt'} +  \tilde{B}^\eps_{tt'}) G'( \tilde{A}^\eps_{ts'} +  \tilde{B}^\eps_{ts'})
\Big] ,
\end{eqnarray*}
where $(\tilde{A}^\eps_{tt'} ,  \tilde{B}^\eps_{tt'},  \tilde{A}^\eps_{ts'} ,  \tilde{B}^\eps_{ts'})$ 
is an independent copy of $({A}^\eps_{tt'} , {B}^\eps_{tt'},  {A}^\eps_{ts'} ,  {B}^\eps_{ts'})$.
We can then write
\begin{eqnarray*}
&&\EE\Big[ \big( \EE\big[ F(Z^\eps_{t'}) G'(Z_{s'}^\eps)|{\cal F}_{t} \big] -   \EE\big[ F(Z^\eps_{t'}) G'(Z_{s'}^\eps)  \big]\big)^2 \Big]\\
&&
\leq \|F\|_\infty \|G'\|_\infty
\EE\Big[
\big( F( A^\eps_{tt'} +  \tilde{B}^\eps_{tt'}) G'( A^\eps_{ts'} +  \tilde{B}^\eps_{ts'})
- F( \tilde{A}^\eps_{tt'} +  \tilde{B}^\eps_{tt'}) G'( \tilde{A}^\eps_{ts'} +  \tilde{B}^\eps_{ts'}) \big)^2
\Big]^{1/2}  \\
&& \leq C \Big( \EE \big[ (A^\eps_{tt'}  - \tilde{A}^\eps_{tt'})^2 \big]^{1/2}+
\EE \big[ (A^\eps_{ts'}  - \tilde{A}^\eps_{ts'})^2 \big]^{1/2}\Big) \\
&& \leq 2 C \Big( \EE \big[ (A^\eps_{tt'}  )^2 \big]^{1/2}+
\EE \big[ (A^\eps_{ts'} )^2 \big]^{1/2}\Big) \\ 
&& 
\leq 2 C \Big[  \Big(  \sigma_{\rm ou}^2 \int_{-\infty}^t {\cal K}^\eps(t'-u)^2 du \Big)^{1/2} 
+
 \Big(  \sigma_{\rm ou}^2 \int_{-\infty}^t {\cal K}^\eps(s'-u)^2 du \Big)^{1/2} \Big]   \\ && 
 \leq 4C \sigma^\eps_{t'-t,\infty}  \leq  {C_1}   \big( 1 \wedge (\eps/(t'-t))^{1-H}\big)  , 
\end{eqnarray*}  
where we used Lemma  \ref{lem:7} in the last inequality. 
 Then, using the fact that ${\cal K}\in L^1$, this gives 
\begin{eqnarray*}
\big| {\rm Cov} \big( \sigma_t^\eps\vartheta^\eps_{t} , \sigma_{t'}^\eps \vartheta^\eps_{t'}  \big) \big|
&\leq & C_2 \int_t^{T} ds   |{\cal K}^\eps(s-t)| \int_{t'}^T ds'  |{\cal K}^\eps(s'-t')|  \big( 1 \wedge (\eps/(t'-t)))^{(1-H)/2}\big)\\
&\leq & C_3 \eps \big( 1 \wedge (\eps/(t'-t)))^{(1-H)/2}\big) ,
\end{eqnarray*} 
which proves the third item.
\end{proof}

\begin{lemma}
\label{lem:0}%
For any smooth function $f$ with bounded derivative, 
we have
\begin{equation}
{\rm Var} \big( \EE \big[ f(Z_{t}^\eps) |{\cal F}_0 \big]\big)
\leq  \|f'\|_\infty^2  (\sigma_{t,\infty}^\eps)^2  .
\label{eq:bornvarcond1}
\end{equation}
\end{lemma}
\begin{proof}
The conditional distribution of $Z_t^\eps$ given ${\cal F}_0$ is Gaussian with mean
$$
\EE \big[  Z_t^\eps |{\cal F}_0 \big] = \sigma_{{\rm ou}}  \int_{-\infty}^0 {\cal K}^\eps(t-u) dW_u
$$
and variance
$$
 {\rm Var} \big( Z_t^\eps |{\cal F}_0\big) = (\sigma_{0,t}^\eps)^2 =\sigma_{{\rm ou}}^2 \int_0^{t} {\cal K}^\eps(u)^2 du  .
$$
Therefore
$$
{\rm Var} \big( \EE \big[ f(Z_{t}^\eps) |{\cal F}_0 \big]\big)
=
{\rm Var} \Big( \int_\RR f\big(\EE \big[ Z_{t}^\eps |{\cal F}_0 \big]  +\sigma_{0,t}^\eps z \big) p(z) dz \Big)  ,
$$
where $p(z)$ is the pdf of the standard normal distribution.
The random variable $\EE \big[ Z_{t}^\eps |{\cal F}_0 \big] $ is Gaussian with mean zero and variance
$(\sigma_{t,\infty}^\eps)^2$
so that
\begin{eqnarray*}
\nonumber
{\rm Var} \big( \EE \big[ f(Z_{t}^\eps) |{\cal F}_0 \big]\big)
&=&
 \frac{1}{2} \int_\RR \int_\RR dz dz' p(z) p(z') \int_\RR \int_\RR du du' p(u) p(u') \\
\nonumber &&\times 
\Big[   f\big(\sigma_{t,\infty}^\eps u +\sigma_{0,t}^\eps z \big) - 
f\big(\sigma_{t,\infty}^\eps u' +\sigma_{0,t}^\eps z \big)\Big] \\
\nonumber&& \times
\Big[   f\big(\sigma_{t,\infty}^\eps u +\sigma_{0,t}^\eps z' \big) - 
f\big(\sigma_{t,\infty}^\eps u' +\sigma_{0,t}^\eps z' \big)\Big] \\
\nonumber&\leq & \|f'\|_\infty^2 (\sigma_{t,\infty}^\eps)^2
\frac{1}{2} \int_\RR \int_\RR du du' p(u) p(u')(u-u')^2 \\
& =  & \|f'\|_\infty^2  (\sigma_{t,\infty}^\eps)^2  ,
\end{eqnarray*}
which is the desired result.
\end{proof}

The random term $\phi^\eps_{t}$ defined by (\ref{def:phit}) has the form 
\begin{equation}
\phi_{t}^\eps = 
\EE \Big[    \int_t^T G (Z_s^\eps)  ds \big| {\cal F}_t \Big] ,
\end{equation}
with $G$ defined in (\ref{def:G}). 
\begin{lemma}
\label{lem:3}%
For any $t \leq T$, $\phi_{t}^\eps$ is a zero-mean random variable with standard deviation of order $\eps^{1-H}$:
\begin{equation}
\label{eq:stdev}
\sup_{\eps \in (0,1]}
\sup_{t \in [0,T]}\eps^{2H-2} 
\EE [ (\phi_{t}^\eps)^2]  < \infty .
\end{equation}
\end{lemma}

\begin{proof}
For $t\in [0,T]$  the second moment of $\phi_{t}^\eps$ is:
\begin{eqnarray*}
\EE \big[ (\phi_{t}^\eps )^2\big] &=&  
\EE\Big[  
\EE \Big[    \int_t^T G (Z_s^\eps)  ds \big| {\cal F}_t \Big]^2
\Big]\\
&=&
\int_0^{T-t} ds \int_{0}^{T-t} ds' 
{\rm Cov}\big( \EE \big[ G(Z_s^\eps)|{\cal F}_0\big] ,\EE \big[ G(Z_{s'}^\eps)|{\cal F}_0\big] \big) .
\end{eqnarray*}
We have by Lemma \ref{lem:0}
\begin{eqnarray*}
\EE \big[ (\phi_{t}^\eps )^2\big]  &\leq &  \Big( \int_0^{T-t} ds 
{\rm Var}\big( \EE \big[ G(Z_s^\eps)|{\cal F}_0\big]\big)^{1/2}
\Big)^2\leq  \|G'\|_\infty^2\Big(  \int_0^{T-t} ds  \sigma_{s,\infty}^\eps  \Big)^2 .
\end{eqnarray*}
In view of Lemma \ref{lem:7} we then have
$$
\EE \big[ (\phi_{t}^\eps )^2\big]\leq C_{T} \big( \eps+\eps^{1-H}\big)^2 \leq 4 C_{T} \eps^{2-2H} ,
$$
uniformly in $t  \leq T$ and $\eps \in (0,1]$ for some constant $C_{T}$.
\end{proof}

\begin{lemma}
\label{lem:6}
Let us define for any $t \in [0,T]$: 
\begin{equation}
{\kappa}^\eps_t = \frac{\eps^{1/2}}{2} \int_0^t \big( (\sigma_s^\eps)^2  -\overline{\sigma}^2\big)  ds
   = \eps^{1/2} \int_0^t G(Z^\eps_s) ds
,
\end{equation} 
as in (\ref{def:kappaeps}).
We have   
\begin{equation}
\lim_{\eps\to 0}
\sup_{t\in [0,T]}  \eps^{-1/2} \EE \big[ ( {\kappa}^\eps_t )^2\big]^{1/2} = 0.
\end{equation} 
\end{lemma}
\begin{proof}
Since the expectation $\EE [ G(Z^\eps_0)]=0$, 
we have
$$
\EE \big[ ( {\kappa}^\eps_t )^2\big] =  \eps 
\EE \Big[ \Big(  \int_0^t G(Z^\eps_s) ds \Big)^2\Big]
=
2{\eps} 
 \int_0^t ds (t-s) {\rm Cov} \big( G(Z^\eps_s) , G(Z^\eps_0) \big) ds .
$$
We have  moreover
\begin{eqnarray*}
\big| {\rm Cov} \big( G(Z^\eps_s) , G(Z^\eps_0) \big)\big|
&=&
\big| \EE\big[ \big( \EE[G(Z^\eps_s)|{\cal F}_0] -\EE[G(Z^\eps_s)]\big) G(Z^\eps_0) \big]\big|\\
&\leq &
\|G\|_\infty
{\rm Var} \big( \EE[G(Z^\eps_s)|{\cal F}_0] \big)^{1/2} .
\end{eqnarray*}
By Lemma \ref{lem:0} we obtain
\begin{eqnarray*}
\big| {\rm Cov} \big( G(Z^\eps_s) , G(Z^\eps_0) \big)\big|
&\leq &
\|G\|_\infty\|G'\|_\infty
\sigma^\eps_{s,\infty} .
\end{eqnarray*}
In view of Lemma \ref{lem:7} we then have
$$
\EE \big[ ( {\kappa}^\eps_t )^2\big] \leq C_T \eps \big(\eps +  \eps^{1-H} \big) 
\leq 2 C_T \eps^{2-H} ,
$$
uniformly in $t \in[0,T]$ and $\eps \in (0,1]$,
which gives the desired result.
\end{proof}

\begin{lemma}
\label{lem:7}%
Define 
\begin{eqnarray}\label{eq:sigi}
\sigma^\eps_{t,\infty} =  \sigma_{\rm ou} \left( \int_t^\infty  {\cal K}^\eps(s)^2 ds \right)^{1/2}   , 
\end{eqnarray}
Then there exists $C>0$ such that
\begin{eqnarray}\label{eq:sigi2}
\sigma^\eps_{t,\infty} \leq C \big( 1 \wedge (\eps/t)^{1-H}\big) .
\end{eqnarray}
\end{lemma}
\begin{proof}
This follows from
  $|{\cal K}(s)| \leq K s^{H-\frac{3}{2}}$ for $s \geq 1$ and ${\cal K}\in L^2$.
\end{proof}

\section{An Alternative Model}\label{sec:appB}
In \cite{comte,funahashi} the authors consider a stochastic volatility model
that is a kind of fractional Orstein-Uhlenbeck process, but
they consider the following representation of the fractional Brownian motion:
\begin{equation}
\label{def:WH0}
W^{H,0}_t=  \frac{1}{\Gamma(H+\frac{1}{2})} \int_0^t (t-s)^{H-\frac{1}{2}}dW^0_s ,
\end{equation}
where $(W^0_t)_{t \in \RR^+}$ is a standard Brownian motion over $\RR^+$.
$(W^{H,0}_t)_{t \in \RR^+}$ is a zero-mean self-similar Gaussian process, in the sense that 
$(\alpha^H W^{H,0}_{t/\alpha })_{t \in \RR^+} $ and $(W_t^{H,0})_{t \in \RR^+}$ have the same distribution,
but it is not stationary,  nor does it have stationary increments.
Its variance is 
$$
\EE\big[ (W_{t}^{H,0})^2 \big]=
 \frac{1}{2H \Gamma(H+\frac{1}{2})^2}  t^{2H},
$$
while the variance of its increment is (for $s>0$)~:
$$
\EE\big[ (W_{t+s}^{H,0}-W_t^{H,0})^2 \big] =
 \frac{1}{\Gamma(H+\frac{1}{2})^2} \Big[ \int_0^{t/s} \big( (1+u)^{H-\frac{1}{2}}-u^{H-\frac{1}{2}}\big)^2 du
 +\frac{1}{2H}\Big] {s}^{2H}  ,
$$
which has the following behavior
$$
\EE\big[ (W_{t+s}^{H,0}-W_t^{H,0})^2 \big] \stackrel{t \to+\infty}{\longrightarrow}
 \frac{1}{\Gamma(2H+1)\sin(\pi H)} {s}^{2H}  .
$$
This model is special because time zero plays a special role, and we think it is 
desirable to deal with the stationary situation addressed in this paper.
However, it turns out that the two models give the same result in the fast-varying case.
Indeed, the modified fOU process corresponding to (\ref{def:WH0})
is (to be compared with~(\ref{eq:fOU})):
\begin{eqnarray}
\nonumber
Z^{\eps,0}_t &=& Z_0 e^{-t/\eps} +\eps^{-H} \int_0^t e^{-\frac{t-s}{\eps}} dW^{H,0}_s\\
&=& Z_0 e^{-t/\eps} +\eps^{-H} W^{H,0}_t - \eps^{-H-1} 
\int_0^t  e^{-\frac{t-s}{\eps}} W^{H,0}_s ds ,
\end{eqnarray}
where $Z_0$ is considered as a constant as in \cite{comte,funahashi}.
In terms of the Brownian motion $W^0_t$ this reads:
\begin{equation}
Z^{\eps,0}_t =Z_0 e^{-t/\eps}+ \sigma_{{\rm ou}} \int_0^t {\cal K}^\eps(t-s) dW^0_s,
\end{equation}
where ${\cal K}^\eps$ is defined in (\ref{def:Keps}).
It is a Gaussian process with the following covariance ($t,s\geq0$):
$$
{\rm Cov} \big( Z^{\eps,0}_t , Z^{\eps,0}_{t+s}  \big) =
\sigma^2_{{\rm ou}} {\cal C}_{t/\eps}^0 \Big(\frac{s}{\eps}\Big) ,
$$
that is a function of $t/\eps$ and $s/\eps$ with
$$
{\cal C}_t^0(s) = 
\frac{\int_0^t {\cal K}(u) {\cal K}(u+s)du}{\int_0^\infty {\cal K}(u)^2du}.
$$
Note that
$$
{\cal C}_t^0(s) \stackrel{t \to +\infty}{\longrightarrow} \frac{\int_0^\infty {\cal K}(u) {\cal K}(u+s)du}{\int_0^\infty {\cal K}(u)^2du} = 
{\cal C}_Z(s) ,
$$
with ${\cal C}_Z$ defined by (\ref{def:calCZ}).
In other words, except for a small period of time just after time $0$ which is of duration of the order of $\eps$, 
the modified process has the same behavior as the one introduced in this paper.
One can then check the detailed calculations carried out in this paper 
and find that Proposition \ref{prop:main}
still holds true with the modified model $Z^{\eps,0}_t$.

\end{document}